\documentclass[11pt,letterpaper]{article}
\usepackage[T1]{fontenc}
\usepackage[utf8]{inputenc}
\usepackage{amsmath,amssymb,amsthm}
\usepackage[colorinlistoftodos]{todonotes}
\usepackage{enumerate}
\usepackage[ruled,vlined]{algorithm2e}
\usepackage{color}

\usepackage{geometry}
\newcommand{\noun}[1]{$\textsc{#1}$}

\definecolor{darkgreen}{rgb}{0,0.8,0}

\renewcommand{\vec}{\mathaccent"017E }

\def\eps{\varepsilon}
\def\ie{{i.\,e.}}
\def\eg{{e.\,g.}}
\def\SSPA{Successive Shortest Path Algorithm}
\def\SM{Simplex Method}
\def\SA{Simplex Algorithm}
\def\NSM{Network \SM}
\def\NSA{Network \SA}

\newtheorem{theorem}{Theorem}
\newtheorem{lemma}{Lemma}

\newtheorem{corollary}{Corollary}

\newtheorem{definition}{Definition}
\newtheorem{proposition}{Proposition}

\newcommand{\reducespace}{\vspace{-2ex}} 

\tikzstyle{node}=[circle, inner sep = 0pt, minimum size = 0.9em, fill]
\tikzstyle{arc}=[->,very thick]

\begin{document}

\title{The Simplex Algorithm is NP-mighty}

\author{
  Yann Disser \and
  Martin Skutella
}


\maketitle


\begin{abstract}
We propose to classify the power of algorithms by the complexity of the problems that they can be used to solve. Instead of restricting to the problem a particular algorithm was designed to solve \emph{explicitly}, however, we include problems that, with polynomial overhead, can be solved `\emph{implicitly}' during the algorithm's execution. For example, we allow to solve a decision problem by suitably transforming the input, executing the algorithm, and observing whether a specific bit in its internal configuration ever switches during the execution. 

We show that the \SM{}, the \NSM{} (both with Dantzig's original pivot rule), and the \SSPA{} are NP-mighty, that is, each of these algorithms can be used to solve any problem in~NP. This result casts a more favorable light on these algorithms' exponential worst-case running times. Furthermore, as a consequence of our approach, we obtain several novel hardness results. For example, for a given input to the \SA, deciding whether a given variable ever enters the basis during the algorithm's execution and determining the number of iterations needed are both NP-hard problems. Finally, we close a long-standing open problem in the area of network flows over time by showing that earliest arrival flows are NP-hard to obtain.
\end{abstract}

\section{Introduction}

Understanding the complexity of algorithmic problems is a central challenge in the theory of computing. Traditionally, complexity theory operates from the point of view of the problems we encounter in the world, by considering a fixed problem and asking how \emph{nice} an algorithm the problem admits with respect to running time, memory consumption, robustness to uncertainty in the input, determinism, etc. In this paper we advocate a different perspective by considering a particular algorithm and asking how powerful (or \emph{mighty}) the algorithm is, \ie, what the most difficult problems are that the algorithm can be used to solve `implicitly' during its execution.

\reducespace
\paragraph{\emph{Related literature.}} A traditional approach to capturing the mightiness of an algorithm is to ask how difficult the exact problem is that the algorithm was designed to solve, \ie, what is the complexity of predicting the final output of the algorithm. For optimization problems, however, if there are multiple optimum solutions to an instance, predicting which optimum solution a specific algorithm will produce might be more difficult than finding an optimum solution in the first place. If this is the case, the algorithm can be considered to be mightier than the problem it is solving suggests. A prominent example for this phenomenon are search algorithms for problems in the complexity class \noun{PLS} (for \emph{polynomial local search}), introduced by Johnson, Papadimitriou, and Yannakakis~\cite{PapadimitriouEtAl88}. Many problems in \noun{PLS} are complete with respect to so-called \emph{tight} reductions, which implies that finding any optimum solution reachable from a specific starting solution via local search is \noun{PSPACE}-complete~\cite{PapadimitriouEtAl90}. Any local search algorithm for such a problem can thus be considered to be \noun{PSPACE}-mighty. Recently, in a remarkable paper by Goldberg, Papadimitriou, and Savani~\cite{GoldbergEtAl13}, similar \noun{PSPACE}-completeness results were established for algorithms solving search problems in the complexity class \noun{PPAD} (for \emph{polynomial parity argument in directed graphs}~\cite{Papadimitriou94}), and in particular for the well-known Lemke-Howson algorithm~\cite{LemkeHowson64} for finding Nash equilibria in bimatrix games.

\reducespace
\paragraph{\emph{A novel approach.}}
We take the analysis of the power of algorithms one step further and argue that the mightiness of an algorithm should not only be classified by the complexity of the exact problem the algorithm is solving, but rather by the most complex problem that the algorithm can be made to solve \emph{implicitly}. In particular, we do not consider the algorithm as a pure black box that turns a given input into a well-defined output. Instead, we are interested in the entire process of computation (\ie, the sequence of the algorithm's internal states) that leads to the final output, and ask how meaningful this process is in terms of valuable information that can be drawn from it. As we show in this paper, sometimes very limited information on an algorithm's process of computation can be used to solve problems that are considerably more complex than the problem the algorithm was actually designed for.

We define the mightiness of an algorithm via the problem of greatest complexity that it can solve \emph{implicitly} in this way, and, in particular, we say that an algorithm is NP-mighty if it implicitly solves all problems in NP (precise definitions are given below). Note that in order to make mightiness a meaningful concept, we need to make sure that mindless exponential algorithms like simple counters do not qualify as being NP-mighty, while algorithms that explicitly solve hard problems do. This goal is achieved by carefully restricting the allowed computational overhead as well as the access to the algorithm's process of computation.

\reducespace
\paragraph{\emph{Considered algorithms.}}
For an algorithm's mightiness to lie beyond the complexity class of the problem it was designed to solve, its running time must be excessive for this complexity class. Most algorithms that are inefficient in this sense would quickly be disregarded as wasteful and not meriting further investigation. Dantzig's \SM{}~\cite{Dantzig-Simplex51} is a famous exception to this rule. Empirically it belongs to the most efficient methods for solving linear programs. However, Klee and Minty~\cite{KleeMinty72} showed that the \SA{} with Dantzig's original pivot rule exhibits exponential worst-case behavior. Similar results are known for many other popular pivot rules; see, \eg, Amenta and Ziegler~\cite{AmentaZiegler96}. 
On the other hand, by the work of Khachiyan~\cite{Khachiyan79,Khachiyan80} and later Karmarkar~\cite{Karmarkar84}, it is known that linear programs can be solved in polynomial time. Spielman and Teng~\cite{SpielmanTeng04} developed the concept of smoothed analysis in order to explain the practical efficiency of the \SM{} despite its poor worst-case behavior.


Minimum-cost flow problems form a class of linear programs featuring a particularly rich combinatorial structure allowing for numerous specialized algorithms. The first such algorithm is Dantzig's \NSM{}~\cite{Dantzig62} which is an interpretation of the general \SM{} applied to this class of problems. In this paper, we consider the primal (Network) \SM{} together with Dantzig's pivot rule, which selects the nonbasic variable with the most negative reduced cost. We refer to this variant of the (Network) \SM{} as the \emph{(Network) \SA{}}. 

One of the simplest and most basic algorithms for minimum-cost flow problems is the \SSPA{} which iteratively augments flow along paths of minimum cost in the residual network~\cite{BusackerGowen60,Iri60}. According to Ford and Fulkerson~\cite{FordFulkerson62}, the underlying theorem stating that such an augmentation step preserves optimality ``\emph{may properly be regarded as the central one concerning minimal cost flows}''. Zadeh~\cite{Zadeh73} presented a family of instances forcing the \SSPA{} and also the \NSA{} into exponentially many iterations. On the other hand, Tardos~\cite{Tardos85} proved that minimum-cost flows can be computed in strongly polynomial time, and Orlin~\cite{Orlin97} gave a polynomial variant of the \NSM{}. 



\reducespace
\paragraph{\emph{Main contribution.}}
We argue that the exponential worst-case running time of the (Network) \SA{} and the \SSPA{} is not purely a waste of time. While these algorithms sometimes take longer than necessary to reach their primary objective (namely to find an optimum solution to a particular linear program), they collect meaningful information on their detours and implicitly solve difficult problems. To make this statement more precise, we introduce a definition of `implicitly solving' that is as minimalistic as possible with regards to the extent in which we are permitted to use the algorithm's internal state. The following definition refers to the \emph{complete configuration} of a Turing machine, \ie, a binary representation of the machine's internal state, contents of its tape, and position of its head.

\begin{definition}\label{def:impl-solve}
An algorithm given by a Turing machine $T$ \emph{implicitly solves} a decision problem $\mathcal{P}$ if, for a given instance $I$ of $\mathcal{P}$, it is possible to compute in polynomial time an input $I'$ for $T$ and a bit~$b$ in the complete configuration of $T$, such that~$I$ is a yes-instance if and only if $b$ flips at some point during the  execution of $T$ for input~$I'$.
\end{definition}


An algorithm that implicitly solves a particular \noun{NP}-hard decision problem implicitly solves all problems in \noun{NP}. We call such algorithms \emph{\noun{NP}-mighty}. 

\begin{definition}\label{def:NP-mighty}
An algorithm is \emph{\noun{NP}-mighty} if it implicitly solves every decision problem in \noun{NP}.
\end{definition}

Note that every algorithm that \emph{explicitly} solves an \noun{NP}-hard decision problem, by definition, also implicitly solves this problem (assuming, w.l.o.g., that a single bit indicates if the Turing machine has reached an accepting state) and thus is \noun{NP}-mighty.

The above definitions turn out to be sufficient for our purposes. We remark, however, that slightly more general versions of Definition~\ref{def:impl-solve}, involving constantly many bits or broader/free access to the algorithm's output, seem reasonable as well. In this context, access to the exact number of iterations needed by the algorithm also seems reasonable as it may provide valuable information. In fact, our results below still hold if the number of iterations is all we may use of an algorithm's behavior. Most importantly, our definitions have been formulated with some care in an attempt to distinguish `clever' exponential-time algorithms from those that rather `waste time' on less meaningful operations. We discuss this critical point in some more detail.

Constructions of exponential time worst-case instances for algorithms usually rely on gadgets that somehow force an algorithm to count, \ie, to enumerate over exponentially many configurations. Such counting behavior by itself cannot be considered `clever', and, consequently, an algorithm should certainly exhibit more elaborate behavior to qualify as being \noun{NP}-mighty. As an example, consider the \emph{simple counting algorithm} (Turing machine) that counts from a given positive number down to zero, \ie, the Turing machine iteratively reduces the binary number on its tape by one until it reaches zero. To show that this algorithm is \emph{not} \noun{NP}-mighty, we need to assume that \noun{P}$\neq$\noun{NP}, as otherwise the polynomial-time transformation of inputs can already solve \noun{NP}-hard problems. Since, for sufficiently large inputs, every state of the simple counting algorithm is reached, and since every bit on its tape flips at some point, our definitions are meaningful in the following sense.

\begin{proposition}
Unless $\textrm{\noun{P}}=\textrm{\noun{NP}}$, the simple counting algorithm is not \noun{NP}-mighty while every algorithm that solves an \noun{NP}-hard problem is \noun{NP}-mighty.
\end{proposition}

Our main result explains the exponential worst-case running time of the following algorithms with their computational power.

\begin{theorem}\label{thm:NP-mighty}
The \SA{}, the \NSA{} (both with Dantzig's pivot rule), and the \SSPA{} are \noun{NP}-mighty.
\end{theorem}

We prove this theorem by showing that the algorithms implicitly solve the \noun{NP}-complete \noun{Partition} problem (cf.~\cite{GareyJohnson79}). To this end, we show how to turn a given instance of \noun{Partition} in polynomial time into a minimum-cost flow network with a distinguished arc~$e$, such that the \NSA{} (or the \SSPA{}) augments flow along arc~$e$ in one of its iterations if and only if the \noun{Partition} instance has a solution. Under the mild assumption that in an implementation of the \NSA{} or the \SSPA{} fixed bits are used to store the flow variables of arcs, this implies that these algorithms implicitly solve \noun{Partition} in terms of Definition~\ref{def:impl-solve}. 

A central part of our network construction is a recursively defined family of counting gadgets on which these minimum-cost flow algorithms take exponentially many iterations. These counting gadgets are, in some sense, simpler than Zadeh's 40 years old `bad networks'~\cite{Zadeh73} and thus interesting in their own right. By slightly perturbing the costs of the arcs according to the values of a given \noun{Partition} instance, we manage to force the considered minimum-cost flow algorithms into enumerating all possible solutions. In contrast to mindless counters, we show that the algorithms are self-aware in the sense that whether or not they encountered a valid \noun{Partition} solution is reflected in their internal state (in the sense of Definition~\ref{def:impl-solve}).

\reducespace
\paragraph{\emph{Further results.}}
We mention interesting consequences of our main results discussed above. Proofs of the following corollaries can be found in Appendix~\ref{sec:corollaries}. We first state complexity results that follow from our proof of Theorem~\ref{thm:NP-mighty}. 

\begin{corollary}\label{cor:ssp_iterations}
Determining the number of iterations needed by the \SA{}, the \NSA{}, and the \SSPA{} for a given input is \noun{NP}-hard.
\end{corollary}

\begin{corollary}
Deciding for a given linear program whether a given variable ever enters the basis during the execution of the \SA{} is \noun{NP}-hard.
\end{corollary}

Another interesting implication is for parametric flows and, more generally, parametric linear programming.

\begin{corollary}\label{cor:parametric}
Determining whether a parametric minimum-cost flow uses a given arc (\ie, assigns positive flow value for any parameter value) is \noun{NP}-hard. In particular, determining whether the solution to a parametric linear program uses a given variable is \noun{NP}-hard. Also, determining the number of different basic solutions over all parameter values is \noun{NP}-hard.
\end{corollary}

We also obtain the following complexity result on $2$-dimensional projections of polyhedra.

\begin{corollary}\label{cor:projection}
Given a $d$-dimensional polytope~$P$ by a system of linear inequalities, determining the number of vertices of $P$'s projection onto a given $2$-dimensional subspace is \noun{NP}-hard.
\end{corollary}

We finally mention a result for a long-standing open problem in the area of network flows over time (see, \eg,~\cite{Skutella-Korte09} for an introduction to this area). The goal in \emph{earliest arrival flows} is to find an $s$-$t$-flow over time that simultaneously maximizes the amount of flow that has reached the sink node~$t$ at any point in time~\cite{Gale59}. It is known since the early 1970ies that the \SSPA{} can be used  to obtain such an earliest arrival flow~\cite{Minieka73,Wilkinson71}. All known encodings of earliest arrival flows, however, suffer from exponential worst-case size, and ever since it has been an open problem whether there is a polynomial encoding which can be found in polynomial time. The following corollary implies that, in a certain sense, earliest arrival flows are \noun{NP}-hard to obtain. 

\begin{corollary}\label{cor:earliest-arrival}
Determining the average arrival time of flow in an earliest arrival flow is \noun{NP}-hard. 
\end{corollary}

Note that an $s$-$t$-flow over time is an earliest arrival flow if and only if it minimizes the average arrival time of flow~\cite{Jarvis82}.

\reducespace
\paragraph{\emph{Outline.}}
After establishing some minimal notation in Section~\ref{sec:preliminaries}, we proceed to proving Theorem~\ref{thm:NP-mighty} for the \SSPA{} in Section~\ref{sec:ssp}. In Section~\ref{sec:ns}, we adapt the construction for the \NSA{}. Finally, Section~\ref{sec:conclusion} highlights interesting open problems for future research. 
All proofs are deferred to the appendix.

\section{Preliminaries}\label{sec:preliminaries}

In the following sections we show that the Successive Shortest
Path Algorithm and the \NSA{} implicitly solve the classical \noun{Partition} problem. An instance of \noun{Partition} is given by
a vector of positive numbers $\vec{a}=(a_{1},\dots,a_{n})\in\mathbb{Q}^{n}$
and the problem is to decide whether there is a subset $I\subseteq\{1,\dots,n\}$
with $\sum_{i\in I}a_{i}=\sum_{i\notin I}a_{i}$. This problem is
well-known to be \noun{NP}-complete (cf.~\cite{GareyJohnson79}). Throughout this paper
we consider an arbitrary fixed instance $\vec{a}$ of \noun{Partition}. Without loss of generality, we assume $A:=\sum_{i=1}^{n}a_{i}<1/12$ and that all values $a_i$, $i\in\{1,\dots,n\}$,
are multiples of $\varepsilon$ for some constant $\varepsilon>0$.

Let $\vec{v}=(v_{1},\dots,v_{n})\in\mathbb{Q}^{n}$ and $k\in\mathbb{N}$,
with $k_{j}\in\{0,1\}$, $j\in \mathbb{Z}_{\geq 0}$, being the $j$-th bit in the
binary representation of $k$, \ie, $k_{j}:=\left\lfloor k/2^{j}\right\rfloor \,\mathrm{mod}\,2$.
We define $\vec{v}_{i_{1},i_{2}}^{[k]}:=\sum_{j=i_{1}+1}^{i_{2}}(-1)^{k_{j-1}}v_{j}$,
$\vec{v}_{i}^{[k]}:=\vec{v}_{0,i}^{[k]}$, and $\vec{v}_{i,i}^{[k]} = 0$. 

The following characterization will be useful later.
\begin{proposition}
\label{prop:ak}The \noun{Partition} instance $\vec{a}$ admits a
solution if and only if there is a $k\in\{0,\dots,2^{n}-1\}$ for which
$\vec{a}_{n}^{[k]}=0$. 
\end{proposition}

\section{\SSPA}\label{sec:ssp}

Consider a network~$N$ with a source node~$s$, a sink node~$t$, and non-negative arc costs. The \SSPA{} starts with the zero-flow and iteratively augments flow along a minimum-cost $s$-$t$-path in the current residual network, until a maximum $s$-$t$-flow has been found. Notice that the residual network is a sub-network of $N$'s bidirected network, where the cost of a backward arc is the negative of the cost of the corresponding forward arc.

\subsection{A Counting gadget for the \SSPA\label{sub:ssp_gadget}}

In this section we construct a family of networks for which the \SSPA{} takes an exponential number of iterations.
Assume we have a network $N_{i-1}$ with source $s_{i-1}$ and sink $t_{i-1}$ which requires $2^{i-1}$ iterations that each augment one unit of flow.
We can obtain a new network $N_{i}$ with only two additional nodes $s_i$, $t_i$ for which the \SSPA{} takes $2^i$ iterations.
To do this we add two arcs $(s_i,s_{i-1})$, $(t_{i-1},t_i)$ with capacity $2^{i-1}$ and cost $0$, and two arcs $(s_i,t_{i-1})$, $(s_{i-1},t_i)$ with capacity $2^{i-1}$ and very high cost.
The idea is that in the first $2^{i-1}$ iterations one unit of flow is routed along the arcs of cost $0$ and through $N_{i-1}$. 
After $2^{i-1}$ iterations both the arcs $(s_i,s_{i-1})$, $(t_{i-1},t_i)$ and the subnetwork $N_{i-1}$ are completely saturated and the \SSPA{} starts to use the expensive arcs $(s_i,t_{i-1})$, $(s_{i-1},t_i)$.
Each of the next $2^{i-1}$ iteration adds one unit of flow along the expensive arcs and removes one unit of flow from the subnetwork $N_{i-1}$.

We tune the cost of the expensive arcs to $2^{i-1}-\frac{1}{2}$ which turns out to be just expensive enough (cf.~Figure~\ref{fig:ssp_gadget}, with $v_i=0$).
This leads to a particularly nice progression of the costs of shortest paths, where the shortest path in iteration $j=0,1,\dots,2^i-1$ simply has cost $j$.
\begin{figure}[t]
\begin{centering}
\begin{tikzpicture}
  \node (sn) at (-0.6,0) [node,label=left:$s_0$] {}; 
  \node (tn) at (0.6,0) [node,label=right:$t_0$] {}; 
  \draw [arc] (sn) to node [auto,sloped,above] {$0$; $1$} (tn); 
  \node at (0,-1) {$N^{\vec{v}}_0$};

  \begin{scope}[xshift=6.5cm] 
   \node (s2) at (-4,0) [node,label=left:$s_i$] {};
   \node (t2) at (4,0) [node,label=right:$t_i$] {};
   \node (s1) at (0,2) [node,label=above:$s_{i-1}$] {};
   \node (t1) at (0,-2) [node,label=below:$t_{i-1}$] {};
   \node (s0) at (0.5,1) {};	 
   \node (s00) at (0.5,-1) {};
   \node (t0) at (-0.5,1) {};	 
   \node (t00) at (-0.5,-1) {};
   \draw [arc] (s2) to node [auto,sloped,above] {$\frac{1}{2}v_i$; $2^{i-1}$} (s1);
   \draw [arc] (s2) to node [auto,sloped,below] {\small$\frac{1}{2}(2^i-1 - v_i)$; $2^{i-1}$~~~} (t1);
   \draw [arc] (s1) to (s0);
   \draw [arc] (s1) to (t0); 
   \draw [arc] (s00) to (t1);
   \draw [arc] (t00) to (t1);
   \draw [arc] (s1) to node [auto,sloped,above] {\small~~~$\frac{1}{2}(2^i-1 - v_i)$; $2^{i-1}$} (t2);
   \draw [arc] (t1) to node [auto,sloped,below] {$\frac{1}{2}v_i$; $2^{i-1}$} (t2);
   \draw (0,0) ellipse (1cm and 3cm); \node at (0,0) {$N^{\vec{v}}_{i-1}$};
   \node at (-3,-2) {$N^{\vec{v}}_{i}$};
  \end{scope} 
\end{tikzpicture}
\par\end{centering}

\caption{\label{fig:ssp_gadget}
Recursive definition of the counting gadget $N_i^{\vec{v}}$ for the Successive Shortest Path Algorithm and $\vec{v}\in\{\vec{a},-\vec{a}\}$. Arcs are labeled by their cost and capacity in this order. The cost of the shortest $s_i$-$t_i$-path in iteration $j=0,\dots,2^i-1$ is $j+\vec{v}_i^{[j]}$.
}
\end{figure}
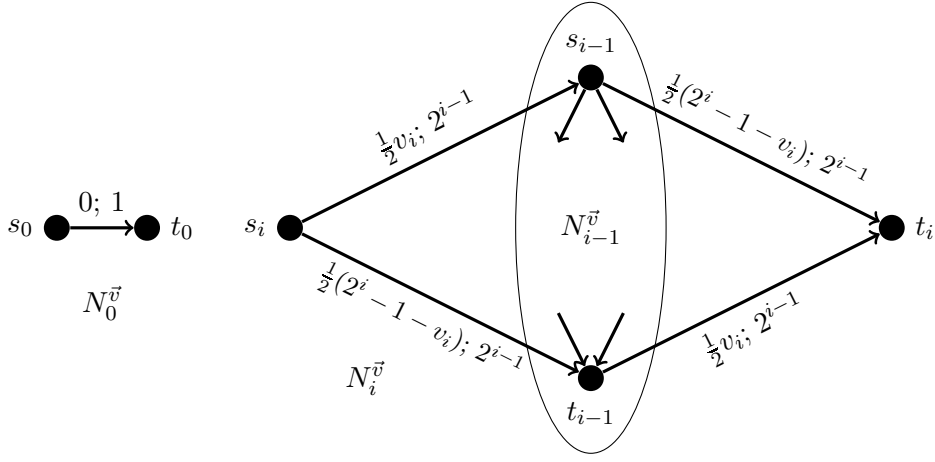

Our goal is to use this counting gadget to iterate over all candidate solutions for a \noun{Partition} instance $\vec{v}$ (we later use the gadget for $\vec{v}\in\{\vec{a},-\vec{a}\}$).
Motivated by Proposition~\ref{prop:ak}, we perturb the costs of the arcs in such a way that the shortest path in iteration $j$ has cost $j + \vec{v}_i^{[j]}$.
We achieve this by adding $\frac{1}{2}v_i$ to the cheap arcs $(s_i,s_{i-1})$, $(t_{i-1},t_i)$ and subtracting $\frac{1}{2}v_i$ from the expensive arcs $(s_i,t_{i-1})$, $(s_{i-1},t_i)$.
If the value of $v_i$ is small enough, this modification does not affect the overall behavior of the gadget.
The first $2^{i-1}$ iterations now have an additional cost of $v_i$ while the next $2^{i-1}$ iterations have an additional cost of $-v_i$, which leads to the desired cost when the modification is applied recursively.

Figure~\ref{fig:ssp_gadget} shows the recursive construction of our counting gadget $N_n^{\vec{v}}$ that encodes the \noun{Partition} instance~$\vec{v}$.
The following lemma formally establishes the crucial properties of the construction.

\begin{lemma}
\label{lem:ssp_gadget}For $\vec{v}\in\{\vec{a},-\vec{a}\}$ and $i=1,\dots,n$,
the \SSPA{} applied to network $N_{i}^{\vec{v}}$
with source $s_{i}$ and sink $t_{i}$ needs $2^{i}$ iterations to
find a maximum $s_{i}$-$t_{i}$-flow of minimum cost. In each iteration
$j=0,1,\dots,2^{i}-1$, the algorithm augments one unit of flow along
a path of cost $j+\vec{v}_{i}^{[j]}$ in the residual network.
\end{lemma}

\subsection{The \SSPA{} implicitly solves \noun{Partition}\label{sub:ssp_construction}}

We use the counting gadget of the previous section to prove Theorem~\ref{thm:NP-mighty} for the \SSPA{}.
Let $G_{\mathrm{ssp}}^{\vec{a}}$ be the network consisting of the
two gadgets $N_{n}^{\vec{a}}$, $N_{n}^{-\vec{a}}$, connected to a new source node $s$ and a new sink $t$ (cf.~Figure~\ref{fig:ssp_construction}).
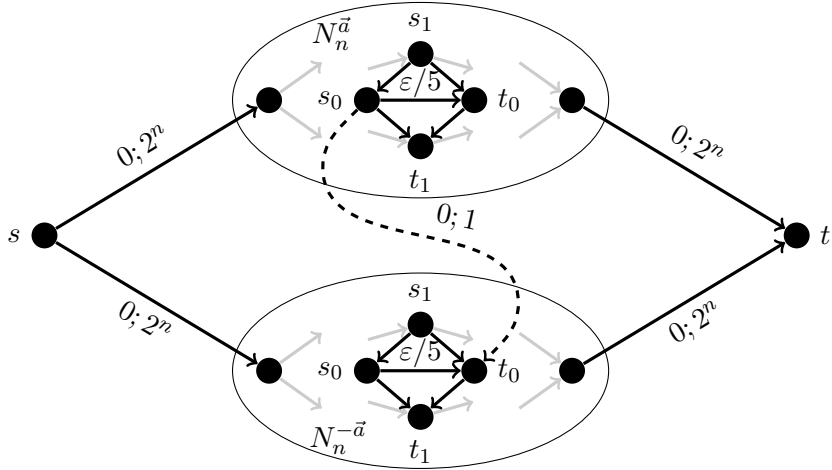
\begin{figure}[t]
\centering
\begin{tikzpicture}[node distance=1cm]
  \node (s) at (-5, 0) [node,label=left:$s$] {}; 
  \node (t) at ( 5, 0) [node,label=right:$t$] {};
  
  \node (Sp) at (0,1.8) {};
  \node (SpName) [above left = 0.4 and 0.65 of Sp] {$N^{\vec{a}}_n$};
  \draw (Sp) ellipse (2.5 and 1.3);
  
  \node (SpIn) [node, left = 1.7 of Sp] {};
  \draw [arc,->,color=black!20] (SpIn) to ++(.7,0.5);
  \draw [arc,->,color=black!20] (SpIn) to ++(.7,-0.5);
  \node (SpOut) [node, right = 1.7 of Sp] {};
  \draw [arc,<-,color=black!20] (SpOut) to ++(-.7,0.5);
  \draw [arc,<-,color=black!20] (SpOut) to ++(-.7,-0.5);
  \draw [arc] (s) to node [auto,sloped,above] {$0;2^n$} (SpIn); 
  \draw [arc] (SpOut) to node [auto,sloped,above] {$0;2^n$} (t);
  
  \node (Sp_s0) [node, left = .4 of Sp,label=left:$s_0$] {};
  \node (Sp_t0) [node, right = .4 of Sp,label=right:$t_0$] {};
  \node (Sp_s1) [node, above = .3 of Sp,label=above:$s_1$] {};
  \node (Sp_t1) [node, below = .3 of Sp,label=below:$t_1$] {};
  
  \draw [arc,color=black!20] (Sp_s1) to ++(.7,-0.2);
  \draw [arc,<-,color=black!20] (Sp_s1) to ++(-.7,-0.2);
  \draw [arc,color=black!20] (Sp_t1) to ++(.7,0.2);
  \draw [arc,<-,color=black!20] (Sp_t1) to ++(-.7,0.2);
  
  \draw [arc] (Sp_s1) to (Sp_s0); \draw [arc] (Sp_s1) to (Sp_t0);
  \draw [arc] (Sp_t0) to (Sp_t1); \draw [arc] (Sp_s0) to (Sp_t1);
  \draw [arc] (Sp_s0) to node [auto,above=-3pt] {$\eps/5$} (Sp_t0);
  
  \node (Sm) at (0,-1.8) {};
  \node (SmName) [below left = 0.4 and 0.45 of Sm] {$N^{-\vec{a}}_n$};
  \draw (Sm) ellipse (2.5 and 1.3);
  
  \node (SmIn) [node, left = 1.7 of Sm] {};
  \draw [arc,->,color=black!20] (SmIn) to ++(.7,0.5);
  \draw [arc,->,color=black!20] (SmIn) to ++(.7,-0.5);
  \node (SmOut) [node, right = 1.7 of Sm] {};
  \draw [arc,<-,color=black!20] (SmOut) to ++(-.7,0.5);
  \draw [arc,<-,color=black!20] (SmOut) to ++(-.7,-0.5);
  \draw [arc] (s) to node [auto,sloped,below] {$0;2^n$} (SmIn); 
  \draw [arc] (SmOut) to node [auto,sloped,below] {$0;2^n$} (t);
  
  \node (Sm_s0) [node, left = .4 of Sm,label=left:$s_0$] {};
  \node (Sm_t0) [node, right = .4 of Sm,label=right:$t_0$] {};
  \node (Sm_s1) [node, above = .3 of Sm,label=above:$s_1$] {};
  \node (Sm_t1) [node, below = .3 of Sm,label=below:$t_1$] {};
  
  \draw [arc,color=black!20] (Sm_s1) to ++(.7,-0.2);
  \draw [arc,<-,color=black!20] (Sm_s1) to ++(-.7,-0.2);
  \draw [arc,color=black!20] (Sm_t1) to ++(.7,0.2);
  \draw [arc,<-,color=black!20] (Sm_t1) to ++(-.7,0.2);
  
  \draw [arc] (Sm_s1) to (Sm_s0); \draw [arc] (Sm_s1) to (Sm_t0);
  \draw [arc] (Sm_t0) to (Sm_t1); \draw [arc] (Sm_s0) to (Sm_t1);
  \draw [arc] (Sm_s0) to node [auto,above=-3pt] {$\eps/5$} (Sm_t0);
  
  \draw [arc,dashed] (Sp_s0) to [out=-135,in=90] (-1.3,.9) to [out=-90,in=90] node [auto,sloped] {$0;1$} (1.3,-.9) to [out=-90,in=45] (Sm_t0);
  
\end{tikzpicture}
  \caption{\label{fig:ssp_construction}
  Illustration of network $G_{\mathrm{ssp}}^{\vec{a}}$. The subnetworks $N_n^{\vec{a}}$ and $N_n^{-\vec{a}}$ are advanced independently by the \SSPA{} without using arc $e$, unless the \noun{Partition} instance $\vec{a}$ has a solution. 
}
\end{figure}
For both of the gadgets, we add the arcs $(s,s_n)$ and $(t_n,t)$ with capacity $2^n$ and cost $0$.
We introduce one additional arc $e$ (dashed in the figure) of capacity $1$ and cost $0$ from node~$s_0$ of gadget $N_{n}^{\vec{a}}$ to node $t_0$ of gadget $N_{n}^{-\vec{a}}$.
Finally, we increase the costs of the arcs $(s_0,t_0)$ in both gadgets from $0$ to $\frac15 \varepsilon$.
Recall that $\varepsilon>0$ is related to $\vec{a}$ by the fact that all $a_{i}$'s are multiples of $\varepsilon$, \ie, a cost smaller than $\varepsilon$ is insignificant compared to all other costs.

\begin{lemma}\label{lem:ssp_construction}
 The \SSPA{} on network $G_{\mathrm{ssp}}^{\vec{a}}$ augments flow along arc $e$ if and only if the \noun{Partition} instance $\vec{a}$ has a solution.
\end{lemma}

We assume that a single bit of complete configuration of the Turing machine corresponding to the \SSPA{} can be used to distinguish whether arc~$e$ carries a flow of $0$ or a flow of $1$ during the execution of the algorithm and that the identity of this bit can be determined in polynomial time. Under this natural assumption, we get the following result, which implies Theorem~\ref{thm:NP-mighty} for the \SSPA{}.

\begin{corollary}\label{cor:ssp_solves_partition}
The \SSPA{} solves \noun{Partition} implicitly.
\end{corollary}

\section{\SA{} and \NSA{}}\label{sec:ns}

In this section we adapt our construction for the \SA{} and, in particular, for its interpretation for the minimum-cost flow problem, the \NSA{}. In this specialized version of the \SA, a basic feasible solution is specified by a spanning tree~$T$ such that the flow value on each arc of the network not contained in~$T$ is either zero or equal to its capacity. We refer to this tree simply as the basis or the spanning tree. The reduced cost of a residual non-tree arc~$e$ equals the cost of sending one unit of flow in the direction of~$e$ around the unique cycle obtained by adding~$e$ to~$T$. For a pair of nodes, the unique path connecting these nodes in the spanning tree~$T$ is referred to as the \emph{tree-path} between the two nodes. Note that while we setup the initial basis and flow manually in the constructions of the following sections, determining the initial feasible flow algorithmically via the algorithm of Edmonds and Karp, ignoring arc costs, yields the same result.

Our construction ensures that all intermediate solutions of the \NSA{} are non-degenerate. Moreover, in every iteration there is a unique non-tree arc of minimum reduced cost which is used as a pivot element.

\subsection{A Counting gadget for the \NSA{}\label{sub:CountingNS}}

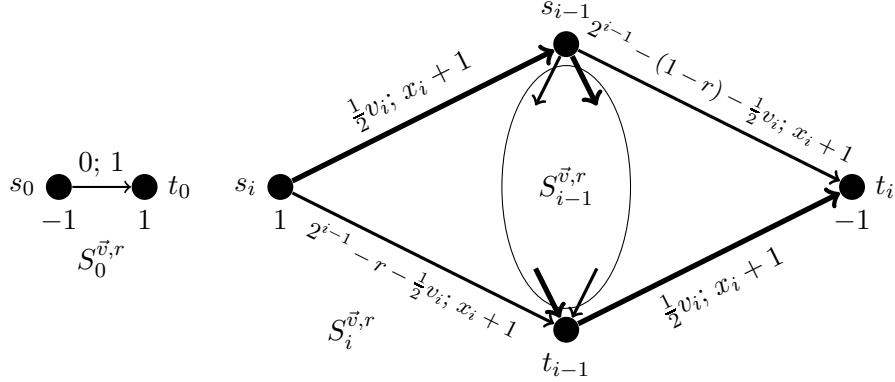
\begin{figure}[t]
\begin{centering}
\begin{tikzpicture}[scale=.95]
  \node (sn) at (-.6,0) [node,label=left:$s_0$,label=below:$-1$] {}; 
  \node (tn) at (.6,0) [node,label=right:$t_0$,label=below:$1$] {}; 
  \draw [arc,thick] (sn) to node [auto,sloped,above] {$0$; $1$} (tn); 
  \node at (0,-1) {$S^{\vec{v},r}_0$};

  \begin{scope}[xshift=6.5cm] \node (s2) at (-4.,0) [node,label=left:$s_i$,label=below:$1$] {};
   \node (t2) at (4.,0) [node,label=right:$t_i$,label=below:$-1$] {};
   \node (s1) at (0,2) [node,label=above:$s_{i-1}$] {};
   \node (t1) at (0,-2) [node,label=below:$t_{i-1}$] {};
   \node (s0) at (0.5,1) {};	 
   \node (s00) at (0.5,-1) {};
   \node (t0) at (-0.5,1) {};	 
   \node (t00) at (-0.5,-1) {};
   \draw [arc,line width=2] (s2) to node [auto,sloped,above] {$\frac{1}{2} v_i$; $x_i+1$} (s1);
   \draw [arc] (s2) to node [auto,sloped,below] {\smaller $2^{i-1} - r - \frac{1}{2}v_i$; $x_i+1$} (t1);
   \draw [arc,line width=2] (s1) to (s0);
   \draw [arc] (s1) to (t0); 
   \draw [arc] (s00) to (t1);
   \draw [arc,line width=2] (t00) to (t1);
   \draw [arc] (s1) to node [auto,sloped,above] {\smaller $2^{i-1} - (1-r) - \frac{1}{2}v_i$; $x_i+1$} (t2);
   \draw [arc,line width=2] (t1) to node [auto,sloped,below] {$\frac{1}{2}v_i$; $x_i+1$} (t2);
   \draw (0,0) ellipse (0.9cm and 1.7cm); 
   \node at (0,0) {$S^{\vec{v},r}_{i-1}$};
   \node at (-3,-2) {$S^{\vec{v},r}_{i}$};
  \end{scope} 
\end{tikzpicture}
\par\end{centering}

\caption{\label{fig:ns_gadget}
Recursive definition of the counting gadget $S_i^{\vec{v},r}$ for the \NSA{}, $\vec{v}\in\{\vec{a},-\vec{a}\}$, and a parameter $r\in (2A,1-2A)$, $r\neq 1/2$. The capacities of the arcs of $S_i^{\vec{a},r} \setminus S_{i-1}^{\vec{a},r}$ are $x_i+1=3\cdot 2^{i-1}$. If we guarantee that there always exists a tree-path from $t_i$ to $s_i$ with sufficiently negative cost outside of the gadget, the cost of iteration $3k$, $k=0,\dots,2^i-1$, within the gadget is $k+\vec{v}_i^{[k]}$. Bold arcs are in the initial basis and carry a flow of at least $1$ throughout the execution of the algorithm.
}
\end{figure}

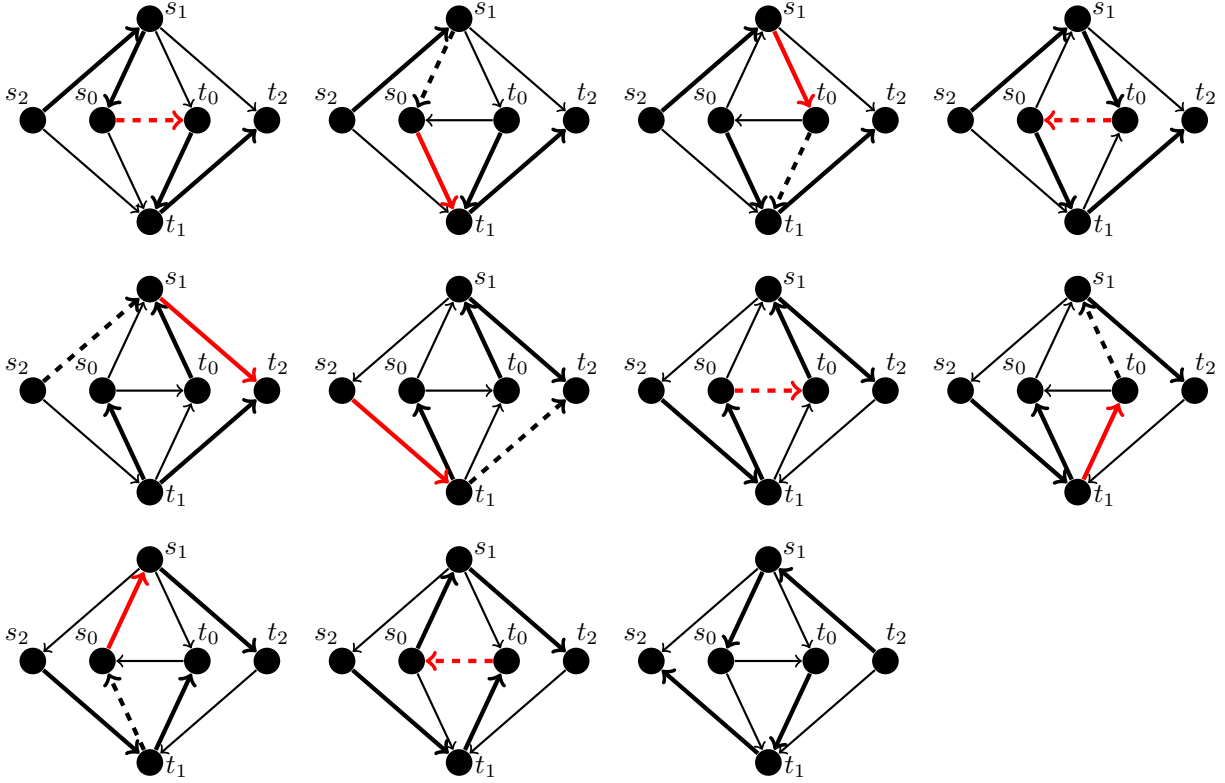
\begin{figure}[t]
\centering
\resizebox{\textwidth}{!}{%
\begin{tikzpicture}
  
  \makeatletter
  \tikzset{%
    prefix node name/.code={%
      \tikzset{%
        name/.code={\edef\tikz@fig@name{#1 ##1}}
      }%
    }%
  }
  \makeatother
  
  \tikzstyle{every node}=[font=\small]
  
  \newcommand\drawnodes
  {
  \node (c) at (0,0) {};
  \node (s0) [node, left  = .3 of c, label={[shift={(-.2,-.1)}]$s_0$}] {};  
  \node (t0) [node, right = .3 of c, label={[shift={(.15,-.1)}]$t_0$}] {};
  \node (s1) [node, above = 1 of c, label={[shift={(.35,-.3)}]$s_1$}] {};   
  \node (t1) [node, below = 1 of c, label={[shift={(.35,-.5)}]$t_1$}] {};
  \node (s2) [node, left  = 1.2 of c, label={[shift={(-.2,-.1)}]$s_2$}] {}; 
  \node (t2) [node, right = 1.2 of c, label={[shift={(.15,-.1)}]$t_2$}] {};
  }  
  \begin{scope}[shift={(0,0)}]
    \drawnodes
    \draw [arc,ultra thick] (s2) to (s1); 
    \draw [arc,ultra thick] (t1) to (t2); 
    \draw [arc,ultra thick] (s1) to (s0); 
    \draw [arc,ultra thick] (t0) to (t1);
    \draw [arc,thick] (s2) to (t1); 
    \draw [arc,thick] (s1) to (t2); 
    \draw [arc,thick] (s1) to (t0); 
    \draw [arc,thick] (s0) to (t1); 
    \draw [arc,ultra thick,dashed,red] (s0) to (t0);
  \end{scope}
  
  \begin{scope}[shift={(4,0)}]
    \drawnodes
    \draw [arc,ultra thick] (s2) to (s1); 
    \draw [arc,ultra thick] (t1) to (t2); 
    \draw [arc,ultra thick] (t0) to (t1);
    \draw [arc,ultra thick,dashed] (s1) to (s0); 
    \draw [arc,thick] (s2) to (t1); 
    \draw [arc,thick] (s1) to (t2); 
    \draw [arc,thick] (s1) to (t0); 
    \draw [arc,thick] (t0) to (s0);
    \draw [arc,ultra thick,red] (s0) to (t1); 
  \end{scope}
  
  \begin{scope}[shift={(8,0)}]
    \drawnodes
    \draw [arc,ultra thick] (s2) to (s1); 
    \draw [arc,ultra thick] (t1) to (t2); 
    \draw [arc,ultra thick] (s0) to (t1); 
    \draw [arc,ultra thick,dashed] (t0) to (t1);
    \draw [arc,thick] (s2) to (t1); 
    \draw [arc,thick] (s1) to (t2); 
    \draw [arc,thick] (t0) to (s0);
    \draw [arc,thick] (s0) to (s1); 
    \draw [arc,ultra thick,red] (s1) to (t0); 
  \end{scope}
  
  \begin{scope}[shift={(12,0)}]
    \drawnodes
    \draw [arc,ultra thick] (s2) to (s1); 
    \draw [arc,ultra thick] (t1) to (t2); 
    \draw [arc,ultra thick] (s0) to (t1); 
    \draw [arc,ultra thick] (s1) to (t0); 
    \draw [arc,thick] (t1) to (t0);
    \draw [arc,thick] (s2) to (t1); 
    \draw [arc,thick] (s1) to (t2); 
    \draw [arc,thick] (s0) to (s1); 
    \draw [arc,ultra thick,dashed,red] (t0) to (s0);
  \end{scope}
  
  \begin{scope}[shift={(0,-3.5)}]
    \drawnodes
    \draw [arc,ultra thick] (t1) to (t2); 
    \draw [arc,ultra thick] (t1) to (s0); 
    \draw [arc,ultra thick] (t0) to (s1); 
    \draw [arc,ultra thick,dashed] (s2) to (s1); 
    \draw [arc,thick] (t1) to (t0);
    \draw [arc,thick] (s2) to (t1); 
    \draw [arc,thick] (s0) to (s1); 
    \draw [arc,thick] (s0) to (t0);
    \draw [arc,ultra thick,red] (s1) to (t2); 
  \end{scope}
  
  \begin{scope}[shift={(4,-3.5)}]
    \drawnodes
    \draw [arc,ultra thick] (t1) to (s0); 
    \draw [arc,ultra thick] (t0) to (s1); 
    \draw [arc,ultra thick] (s1) to (t2); 
    \draw [arc,ultra thick,dashed] (t1) to (t2); 
    \draw [arc,thick] (s1) to (s2); 
    \draw [arc,thick] (t1) to (t0);
    \draw [arc,thick] (s0) to (s1); 
    \draw [arc,thick] (s0) to (t0);
    \draw [arc,ultra thick,red] (s2) to (t1); 
  \end{scope}
  
  \begin{scope}[shift={(8,-3.5)}]
    \drawnodes
    \draw [arc,ultra thick] (t1) to (s0); 
    \draw [arc,ultra thick] (t0) to (s1); 
    \draw [arc,ultra thick] (s1) to (t2); 
    \draw [arc,ultra thick] (s2) to (t1); 
    \draw [arc,thick] (t2) to (t1); 
    \draw [arc,thick] (s1) to (s2); 
    \draw [arc,thick] (t1) to (t0);
    \draw [arc,thick] (s0) to (s1); 
    \draw [arc,ultra thick,dashed,red] (s0) to (t0);
  \end{scope}
  
  \begin{scope}[shift={(12,-3.5)}]
    \drawnodes
    \draw [arc,ultra thick] (t1) to (s0); 
    \draw [arc,ultra thick] (s1) to (t2); 
    \draw [arc,ultra thick] (s2) to (t1); 
    \draw [arc,ultra thick,red] (t1) to (t0);
    \draw [arc,thick] (t2) to (t1); 
    \draw [arc,thick] (s1) to (s2); 
    \draw [arc,thick] (s0) to (s1); 
    \draw [arc,thick] (t0) to (s0);
    \draw [arc,ultra thick,dashed] (t0) to (s1); 
  \end{scope}
  
  \begin{scope}[shift={(0,-7)}]
    \drawnodes
    \draw [arc,ultra thick] (s1) to (t2); 
    \draw [arc,ultra thick] (s2) to (t1); 
    \draw [arc,ultra thick] (t1) to (t0);
    \draw [arc,ultra thick,dashed] (t1) to (s0); 
    \draw [arc,thick] (t2) to (t1); 
    \draw [arc,thick] (s1) to (s2); 
    \draw [arc,thick] (t0) to (s0);
    \draw [arc,thick] (s1) to (t0); 
    \draw [arc,ultra thick,red] (s0) to (s1); 
  \end{scope}
  
  \begin{scope}[shift={(4,-7)}]
    \drawnodes
    \draw [arc,ultra thick] (s1) to (t2); 
    \draw [arc,ultra thick] (s2) to (t1); 
    \draw [arc,ultra thick] (t1) to (t0);
    \draw [arc,ultra thick] (s0) to (s1); 
    \draw [arc,thick] (s0) to (t1); 
    \draw [arc,thick] (t2) to (t1); 
    \draw [arc,thick] (s1) to (s2); 
    \draw [arc,thick] (s1) to (t0); 
    \draw [arc,ultra thick,red,dashed] (t0) to (s0);
  \end{scope}
  
  \begin{scope}[shift={(8,-7)}]
    \drawnodes
    \draw [arc,ultra thick] (t2) to (s1); 
    \draw [arc,ultra thick] (t1) to (s2); 
    \draw [arc,ultra thick] (t0) to (t1);
    \draw [arc,ultra thick] (s1) to (s0); 
    \draw [arc,thick] (s0) to (t1); 
    \draw [arc,thick] (t2) to (t1); 
    \draw [arc,thick] (s1) to (s2); 
    \draw [arc,thick] (s1) to (t0); 
    \draw [arc,thick] (s0) to (t0);
  \end{scope}
\end{tikzpicture}
}
\caption{\label{fig:ssp_example}
Illustration of the iterations performed by the \NSA{} on the counting gadget $S_2^{\vec{a},r}$ for $r<1/2$. The external tree-path from $t_2$ to $s_2$ is not shown. Bold arcs are in the basis before each iteration, the red arc enters the basis and the dashed arc exits the basis. Arcs are oriented in the direction in which they are used next. Note that after $2x_2 = 3 \cdot 2^2 - 2 = 10$ iterations the configuration is the same as in the beginning if we switch the roles of $s_2$ and $t_2$.
}
\end{figure}

We design a counting gadget for the \NSA (cf.~Figure~\ref{fig:ns_gadget}), similar to the gadget $N_i^{\vec{v}}$ of Section~\ref{sub:ssp_gadget} for the \SSPA{}.
Since the \NSA{} augments flow along cycles obtained by adding one arc to the current spanning tree, we assume that the tree always contains an external tree-path from the sink of the gadget to its source with a very low (negative) cost. This assumption will be justified below in Section~\ref{sub:ns_construction} when we embed the counting gadget into a larger network.

The main challenge when adapting the gadget $N_i^{\vec{v}}$ is that the spanning trees in consecutive iterations of the \NSA{} differ in one arc only, since in each iteration a single arc may enter the basis. However, successive shortest paths in $N_i^{\vec{v}}$ differ by exactly two tree-arcs between consecutive iterations.
We obtain a new gadget $S_i^{\vec{v}}$ from $N_i^{\vec{v}}$ by modifying arc capacities in such a way that we get two intermediate iterations for every two successive shortest paths in $N_i^{\vec{v}}$ that serve as a transition between the two paths and their corresponding spanning trees.
Recall that in $N_i^{\vec{v}}$ the capacities of the arcs of $N_i^{\vec{v}} \setminus N_{i-1}^{\vec{v}}$ are exactly the same as the capacity of the subnetwork $N_{i-1}^{\vec{v}}$.
In $S_i^{\vec{v}}$, we increase the capacity by one unit relative to the capacity of $S_{i-1}^{\vec{v}}$. 
The resulting capacities of the arcs in $S_i^{\vec{v}} \setminus S_{i-1}^{\vec{v}}$ are $x_i$ (for the moment), where $x_i = 2x_{i-1} + 1$ and $x_1 = 2$, \ie, $x_i=3 \cdot 2^{i-1} - 1$.

Similar to before, after $2x_{i-1}$ iterations the subnetwork $S_{i-1}^{\vec{v}}$ is saturated.
In contrast however, at this point the arcs $(s_i,s_{i-1})$, $(t_{i-1},t_i)$ are not saturated yet.
Instead, in the next two iterations, the arcs $(s_i,t_{i-1})$, $(s_{i-1},t_i)$ enter the basis and one unit of flow gets sent via the paths $s_i,s_{i-1},t_i$ and $s_i,t_{i-1},t_i$, which saturates the arcs $(s_i,s_{i-1})$, $(t_{i-1},t_i)$ and eliminates them from the basis.
Afterwards, in the next $2x_{i-1}$ iterations, flow is sent via $(s_i,t_{i-1})$, $(s_{i-1},t_i)$ and through $S_{i-1}^{\vec{v}}$ as before (cf.~Figure~\ref{fig:ssp_example} for an example execution of the \NSA{} on $S_2^{\vec{v}}$).

For the construction to work, we need that, in every non-intermediate iteration, arc $(s_0,t_0)$ not only enters the basis but, more importantly, is also the unique arc to leave the basis.
In other words, we want to ensure that no other arc becomes tight in these iterations.
For this purpose, we add an initial flow of~$1$ along the paths $s_i,s_{i-1},\dots,s_0$ and $t_0,t_1,\dots,t_i$ by adding supply 1 to $s_i$, $t_0$ and demand 1 to $s_0$, $t_i$ and increasing the capacities of the affected arcs by~$1$. The arcs of these two paths are the only arcs from the gadget that are contained in the initial spanning tree. We also increase the capacities of the arcs $(s_i,t_{i-1})$, $(s_{i-1},t_i)$ by one to ensure that these arcs are never saturated.

Finally, we also make sure that in every iteration the arc entering the basis is unique.
To achieve this, we introduce a parameter $r\in (2A,1-2A)$, $r\neq1/2$ and replace the costs of $2^{i-1} - \frac12 - \frac12v_i$ of the arcs $(s_i,t_{i-1})$, $(s_{i-1},t_i)$ by new costs $2^{i-1} - r - \frac12 v_i$ and $2^{i-1} - (1-r) - \frac12 v_i$, respectively.

We later use the final gadget $S_{n}^{\vec{v},r}$ as part of a larger network $G$ by connecting the nodes $s_{n},t_{n}$ to nodes in $G\setminus S_{n}^{\vec{v},r}$. 
The following lemma establishes the crucial properties of the gadget
used in such a way as a part of a larger network~$G$. 

\begin{lemma}
\label{lem:ns_gadget}Let $S_{i}^{\vec{v},r}$, $\vec{v}\in\{\vec{a},-\vec{a}\}$,
be part of a larger network $G$ and assume that before every iteration
of the \NSA{} on $G$ where flow is routed through $S_{i}^{\vec{v},r}$
there is a tree-path from $t_{i}$ to $s_{i}$ in the residual network
of $G$ that has cost smaller than $-2^{i+1}$ and capacity greater than~$1$. Then, there are exactly $2x_{i}=3\cdot2^{i}-2$ iterations in which
one unit of flow is routed from $s_{i}$ to~$t_{i}$ along arcs of $S_{i}^{\vec{v},r}$.
Moreover:
\begin{enumerate}
\item \label{one}In iteration $j=3k$, $k=0,\dots,2^{i}-1$, arc $(s_{0},t_{0})$ enters
the basis carrying flow $k\,\mathrm{mod}\,2$ and immediately exits the basis again
carrying flow $(k+1)\,\mathrm{mod}\,2$. The cost incurred by arcs of
$S_{i}^{\vec{v},r}$ is $k+\vec{v}_{i}^{[k]}$.
\item In iterations $j=3k+1,3k+2$, $k=0,\dots,2^{i}-2$, for some $0\leq i'\leq i$, the cost
incurred by arcs of $S_{i}^{\vec{v},r}$ is $k+r+\vec{v}_{i',i}^{[k]}$
and $k+(1-r)+\vec{v}_{i',i}^{[k]}$ in order of increasing cost. One
of the arcs $(s_{i'},s_{i'-1}),(s_{i'-1},t_{i'})$ and one of the arcs
$(s_{i'},t_{i'-1}),(t_{i'-1},t_{i'})$ each enter and leave the basis
in these iterations. 
\end{enumerate}
\end{lemma}

\subsection{The \NSA{} implicitly solves \noun{Partition}\label{sub:ns_construction}}

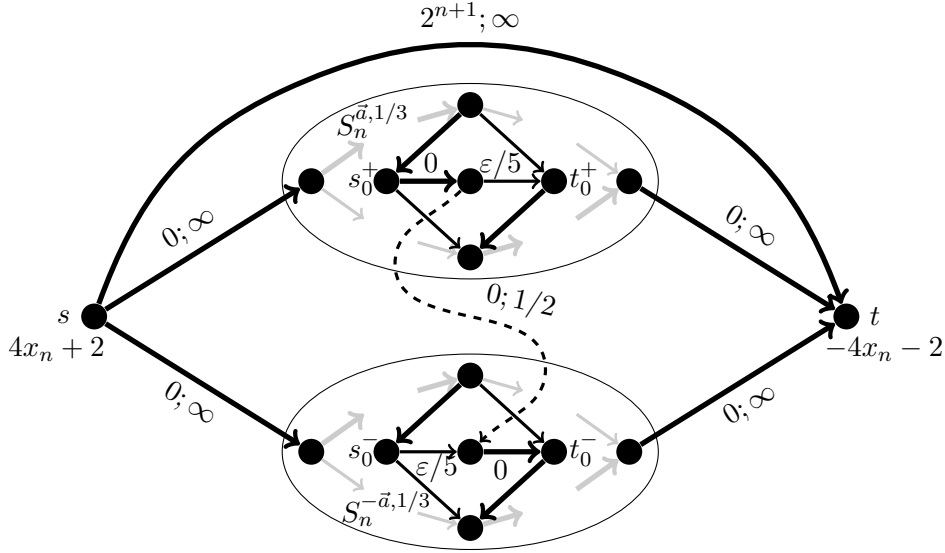
\begin{figure}[t]
\centering
\begin{tikzpicture}[node distance=0cm]
  \node (s) at (-5, 0) [node,label=left:$s$,label={[shift={(-.5,-.9)}]$4x_n + 2$}] {}; 
  \node (t) at ( 5, 0) [node,label=right:$t$,label={[shift={(.5,-.9)}]$-4x_n - 2$}] {};
  
  \node (Sp) at (0,1.8) {};
  \node (SpName) [above left = 0.3 and 0.55 of Sp] {$S^{\vec{a},1/3}_n$};
  \draw (Sp) ellipse (2.5 and 1.3);
  
  \node (SpIn) [node, left = 1.8 of Sp] {};
  \draw [arc,->,color=black!20,line width = 2] (SpIn) to ++(.7,0.5);
  \draw [arc,->,color=black!20] (SpIn) to ++(.7,-0.5);
  \node (SpOut) [node, right = 1.8 of Sp] {};
  \draw [arc,<-,color=black!20,line width = 2] (SpOut) to ++(-.7,-0.5);
  \draw [arc,<-,color=black!20] (SpOut) to ++(-.7,0.5);
  \draw [arc,line width = 2] (s) to node [auto,sloped,above] {$0;\infty$} (SpIn); 
  \draw [arc,line width = 2] (SpOut) to node [auto,sloped,above] {$0;\infty$} (t);
  
  \node (Sp_s0) [node, left = .8 of Sp,label={[shift={(-.3,-0.5)}]$s_0^+$}] {};
  \node (Sp_t0) [node, right = .8 of Sp,label={[shift={(.4,-0.5)}]$t_0^+$}] {};
  \node (Sp_s1) [node, above = .7 of Sp] {};
  \node (Sp_t1) [node, below = .7 of Sp] {};
  \node (Sp_c) [node] at (Sp){};
  
  \draw [arc,color=black!20] (Sp_s1) to ++(.7,-0.2);
  \draw [arc,<-,color=black!20,line width = 2] (Sp_s1) to ++(-.7,-0.2);
  \draw [arc,color=black!20,line width = 2] (Sp_t1) to ++(.7,0.2);
  \draw [arc,<-,color=black!20] (Sp_t1) to ++(-.7,0.2);
  
  \draw [arc,line width = 2] (Sp_s1) to (Sp_s0); \draw [arc] (Sp_s1) to (Sp_t0);
  \draw [arc,line width = 2] (Sp_t0) to (Sp_t1); \draw [arc] (Sp_s0) to (Sp_t1);
  \draw [arc,line width = 2] (Sp_s0) to node [auto,pos=.55,above=-1pt] {$0$} (Sp_c);
  \draw [arc] (Sp_c) to node [auto,pos=.3,above=-3.5pt] {$\eps/5$} (Sp_t0);
  
  \node (Sm) at (0,-1.8) {};
  \node (SmName) [below left = 0.25 and 0.25 of Sm] {$S^{-\vec{a},1/3}_n$};
  \draw (Sm) ellipse (2.5 and 1.3);
  
  \node (SmIn) [node, left = 1.8 of Sm] {};
  \draw [arc,->,color=black!20,line width = 2] (SmIn) to ++(.7,0.5);
  \draw [arc,->,color=black!20] (SmIn) to ++(.7,-0.5);
  \node (SmOut) [node, right = 1.8 of Sm] {};
  \draw [arc,<-,color=black!20,line width = 2] (SmOut) to ++(-.7,-0.5);
  \draw [arc,<-,color=black!20] (SmOut) to ++(-.7,0.5);
  \draw [arc,line width = 2] (s) to node [auto,sloped,below] {$0;\infty$} (SmIn); 
  \draw [arc,line width = 2] (SmOut) to node [auto,sloped,below] {$0;\infty$} (t);
  
  \node (Sm_s0) [node, left = .8 of Sm,label={[shift={(-.3,-0.5)}]$s_0^-$}] {};
  \node (Sm_t0) [node, right = .8 of Sm,label={[shift={(.4,-0.5)}]$t_0^-$}] {};
  \node (Sm_s1) [node, above = .7 of Sm] {};
  \node (Sm_t1) [node, below = .7 of Sm] {};
  \node (Sm_c) [node] at (Sm){};
  
  \draw [arc,color=black!20] (Sm_s1) to ++(.7,-0.2);
  \draw [arc,<-,color=black!20,line width = 2] (Sm_s1) to ++(-.7,-0.2);
  \draw [arc,color=black!20,line width = 2] (Sm_t1) to ++(.7,0.2);
  \draw [arc,<-,color=black!20] (Sm_t1) to ++(-.7,0.2);
  
  \draw [arc,line width = 2] (Sm_s1) to (Sm_s0); \draw [arc] (Sm_s1) to (Sm_t0);
  \draw [arc,line width = 2] (Sm_t0) to (Sm_t1); \draw [arc] (Sm_s0) to (Sm_t1);
  \draw [arc] (Sm_s0) to node [auto,below=-3.5pt,pos=.65] {$\eps/5$} (Sm_c);
  \draw [arc,line width = 2] (Sm_c) to node [auto,below=-1.5pt,pos=.3] {$0$} (Sm_t0);
  
  \draw [arc,dashed] (Sp_c) to [out=-135,in=90] (-1,.7) to [out=-90,in=90] node [auto,sloped] {$0;1/2$} (1,-.7) to [out=-90,in=45] (Sm_c);
  
  \draw [arc,line width=2] (s) to [out=70,in=205] (-2.5,3) to [out=25,in=155] node [auto,sloped] {$2^{n+1};\infty$} (2.5,3) to [out=-25,in=110] (t);
\end{tikzpicture}
  \caption{\label{fig:ns_construction}
  Illustration of network $G_{\mathrm{ns}}^{\vec{a}}$. The subnetworks $S_n^{\vec{a},1/3}$ and $S_n^{-\vec{a},1/3}$ are advanced independently by the \NSA{} without using the dashed arc $e$, unless the \noun{Partition} instance $\vec{a}$ has a solution. Bold arcs are in the initial basis and carry a flow of at least $1$ throughout the execution of the algorithm.
}
\end{figure}

We construct a network $G_{\mathrm{ns}}^{\vec{a}}$ similar to the network $G_{\mathrm{ssp}}^{\vec{a}}$ of Section~\ref{sub:ssp_construction}.
Without loss of generality, we assume that $a_1 = 0$.
The network $G_{\mathrm{ns}}^{\vec{a}}$ consists of the two gadgets $S_{n}^{\vec{a},1/3}$, $S_{n}^{-\vec{a},1/3}$, connected to a new source node $s$ and a new sink $t$ (cf.~Figure~\ref{fig:ns_construction}). Let $s_i^+$, $t_i^+$ denote the nodes of $S_n^{\vec{a},1/3}$ and $s_i^-$, $t_i^-$ denote the nodes of $S_n^{-\vec{a},1/3}$.
We introduce arcs $(s,s_n^+)$, $(s,s_n^-)$, $(t_n^+,t)$, $(t_n^-,t)$, each with capacity $\infty$ and cost $0$. The supply~$1$ of $s_n^+$ and $s_n^-$ is moved to $s$ and the initial flow on arcs $(s,s_n^+)$ and $(s,s_n^-)$ is set to~$1$. Similarly, the demand~$1$ of $t_n^+$ and $t_n^-$ is moved to $t$ and the initial flow on arcs $(t_n^+,t)$ and $(t_n^-,t)$ is set to~$1$. Finally, we add an infinite capacity arc $(s,t)$ of cost $2^{n+1}$, increase the supply of $s$ and the demand of $t$ by $4x_n$, and set the initial flow on $(s,t)$ to $4x_n$.

In addition, we add two new nodes $c^+$, $c^-$ and replace the arc $(s_0^+,t_0^+)$ by two arcs $(s_0^+,c^+)$, $(c^+,t_0^+)$ of capacity $2$ and cost $0$ (for the moment), and analogously for the arc $(s_0^-,t_0^-)$ and $c^-$.
Finally, we move the demand of $1$ from $s_0^+$ to $c^+$ and the supply of $1$ from $t_0^-$ to $c^-$.
The arcs $(s_0^+,c^+)$ and $(c^-,t_0^-)$ carry an initial flow of $1$ and are part of the initial basis.
Observe that these modifications do not change the behavior of the gadgets.
In addition to the properties of Lemma~\ref{lem:ns_gadget} we have that whenever the arc $(s_0,t_0)$ previously carried a flow of $1$, now the arc $(c^+,t_0^+)$ or $(s_0^-,c^-)$ is in the basis, and whenever $(s_0,t_0)$ previously did not carry flow, now the arc $(s_0^+,c^+)$ or $(c^-,t_0^-)$ is in the basis.

We slightly increase the costs of the arcs $(c^+,t_0^+)$ and $(s_0^-,c^-)$ from $0$ to~$\frac15 \varepsilon$, again without affecting the behavior of the gadgets
 (note that we can perturb all costs in $S^{-\vec{a},1/3}_n$ further to ensure that every pivot step is unique).
Finally, we add one more arc $e = (c^+,c^-)$ with cost $0$ and capacity~$\frac12$. 

\begin{lemma}\label{lem:ns_construction}
Arc $e$ enters the basis in some iteration of the \NSA{} on network $G_{\mathrm{ns}}^{\vec{a}}$ if and only if the \noun{Partition} instance $\vec{a}$ has a solution.
\end{lemma}

Again, we assume that a single bit of the complete configuration of the Turing machine corresponding to the \SA{} can be used to detect whether a variable is in the basis and that the identity of this bit can be determined in polynomial time. Under this natural assumption, we get the following result, which implies Theorem~\ref{thm:NP-mighty} for the \NSA{} and thus the \SA{}.

\begin{corollary}
The \NSA{} implicitly solves \noun{Partition}.
\end{corollary}

\section{Conclusion}\label{sec:conclusion}

We have introduced the concept of \emph{\noun{NP}-mightiness} as a novel means of classifying the computational power of algorithms. Furthermore, we have given a justification for the exponential worst-case behavior of \SSPA{} and the (Network) \SM{} (with Dantzig's pivot rule): These algorithms can implicitly solve any problem in \noun{NP}.

A natural open problem is whether the studied algorithms are perhaps even more powerful than our results suggest. Maybe, similarly to the result of Goldberg et al.~\cite{GoldbergEtAl13} for the Lemke-Howson algorithm, the \SA{} can be shown to implicitly solve even \noun{PSPACE}-hard problems. In line with this question, it would be interesting to investigate how difficult it is to predict which optimum solution the \SA{} will produce for a fixed pivot rule, \ie, how difficult is the problem the \SA{} is \emph{explicitly} solving? 

We hope that our approach will turn out to be useful in developing a better understanding of other algorithms that suffer from poor worst-case behavior. In particular, we believe that our results can be carried over to the \SM{} with other pivot rules. Furthermore, even polynomial-time algorithms with a super-optimal worst-case running time are an interesting subject. Such algorithms might implicitly solve problems that are presumably more difficult than the problem they were designed for. In order to achieve meaningful results in this context, our definition of `implicitly solving' (Definition~\ref{def:impl-solve}) would need to be modified by further restricting the running time of the transformation of instances.





\bibliographystyle{plain}
\bibliography{SuccShortPath.bib}

\newpage

\appendix

\section{Omitted proofs of Section~\ref{sec:ssp}}

\setcounter{lemma}{0}

\begin{lemma}
For $\vec{v}\in\{\vec{a},-\vec{a}\}$ and $i=1,\dots,n$,
the \SSPA{} applied to network $N_{i}^{\vec{v}}$
with source $s_{i}$ and sink $t_{i}$ needs $2^{i}$ iterations to
find a maximum $s_{i}$-$t_{i}$-flow of minimum cost. In each iteration
$j=0,1,\dots2^{i}-1$, the algorithm augments one unit of flow along
a path of cost $j+\vec{v}_{i}^{[j]}$ in the residual network.
\end{lemma}

\begin{proof}
We prove the lemma by induction on $i$, together with the additional
property that after $2^{i}$ iterations none of the arcs in $N_{i-1}^{\vec{v}}$
carries any flow, while the arcs in $N_{i}^{\vec{v}}\setminus N_{i-1}^{\vec{v}}$
are fully saturated. First consider the network $N_{0}^{\vec{v}}$.
In each iteration where $N_{0}^{\vec{v}}$ does not carry flow, one
unit of flow can be routed from $s_{0}$ to $t_{0}$. Conversely,
when $N_{0}^{\vec{v}}$ is saturated, one unit of flow can be routed
from $t_{0}$ to $s_{0}$. In either case the associated cost is 0.
With this in mind, it is clear that on $N_{1}^{\vec{v}}$ the
\SSPA{} terminates after two iterations.
In the first, one unit of flow is sent along the path $s_{1},s_{0},t_{0},t_{1}$
of cost $v_{1}=\vec{v}_{1}^{[0]}$. In the second iteration, one unit
of flow is sent along the path $s_{1},t_{0},s_{0},t_{1}$ of cost
$-v_{1}=\vec{v}_{1}^{[1]}$. Afterwards, the arc $(s_{0},t_{0})$
does not carry any flow, while all other arcs are fully saturated. 

Now assume the claim holds for $N_{i-1}^{\vec{v}}$ and consider
network $N_{i}^{\vec{v}}$, $i>1$. Observe that every path using either
of the arcs $(s_{i},t_{i-1})$ or $(s_{i-1},t_{i})$ has a cost of
more than $2^{i-1}-3/4$. To see this, note that the cost of these
arcs is bounded individually by $\frac{1}{2}(2^{i}-1-v_{i})>2^{i-1}-3/4$,
since $|v_{i}|<A<1/4$. On the other hand, it can be seen inductively that
the shortest $t_{i-1}$-$s_{i-1}$-path in the bidirected network associated
with $N_{i-1}^{\vec{v}}$ has cost at least $-2^{i-1}+1-A>-2^{i-1}+3/4$. Hence,
using both $(s_{i},t_{i-1})$ and $(s_{i-1},t_{i})$ in addition to
a path from $t_{i-1}$ to $s_{i-1}$ incurs cost at least $2^{i-1}-3/4$.
By induction, in every iteration $j<2^{i-1}$, the \SSPA{} thus does not use the arcs $(s_{i},t_{i-1})$ or $(s_{i-1},t_{i})$
but instead augments one unit of flow along the arcs $(s_{i},s_{i-1}),(t_{i-1},t_{i})$
and along an $s_{i-1}$-$t_{i-1}$-path of cost $j+\vec{v}_{i-1}^{[j]}<2^{i-1}-3/4$
through the subnetwork $N_{i-1}^{\vec{v}}$. The total cost of this
$s_i$-$t_i$-path is $v_{i}+(j+\vec{v}_{i-1}^{[j]})=j+\vec{v}_{i}^{[j]}$, since
$j<2^{i-1}$.

After $2^{i-1}$ iterations, the arcs $(s_{i},s_{i-1})$ and $(t_{i-1},t_{i})$
are both fully saturated, as well as (by induction) the arcs in $N_{i-1}^{\vec{v}}\setminus N_{i-2}^{\vec{v}}$,
while all other arcs are without flow. Consider the residual network
of $N_{i-1}^{\vec{v}}$ at this point. If we increase
the costs of the four residual arcs in $N_{i-1}^{\vec{v}}\setminus N_{i-2}^{\vec{v}}$ 
by $\frac{1}{2}(2^{i-1}-1)$ 
and switch the roles of $s_{i-1}$ and $t_{i-1}$, we obtain back
the original subnetwork $N_{i-1}^{\vec{v}}$. The shift of the residual
costs effectively makes every $t_{i-1}$-$s_{i-1}$-path more expensive
by $2^{i-1}-1$, but does not otherwise affect the behavior of the
network. We can thus use induction again to infer that in every iteration
$j=2^{i-1},\dots,2^{i}-1$ the \SSPA{}
augments one unit of flow along a path via $s_{i},t_{i-1},N_{i-1}^{\vec{v}},s_{i-1},t_{i}$.
Accounting for the shift in cost by $2^{i-1}-1$, we obtain that this
path has a total cost of
\[
(2^{i}-1-v_{i})+(j-2^{i-1}+\vec{v}_{i-1}^{[j-2^{i-1}]})-(2^{i-1}-1)=j+\vec{v}_{i}^{[j]},
\]
where we used $\vec{v}_{i-1}^{[j-2^{i-1}]}=\vec{v}_{i-1}^{[j]}$ and
$\vec{v}_{i-1}^{[j]}-v_{i}=\vec{v}_{i}^{[j]}$ for $j\in[2^{i-1},2^{i})$.
After $2^{i}$ iterations the arcs in $N_{i}^{\vec{v}}\setminus N_{i-1}^{\vec{v}}$
are fully saturated and all other arcs carry no flow.
\end{proof}

\setcounter{lemma}{1}

\begin{lemma}
 The \SSPA{} on network $G_{\mathrm{ssp}}^{\vec{a}}$ augments flow along arc $e$ if and only if the \noun{Partition} instance $\vec{a}$ has a solution.
\end{lemma}

\begin{proof}
First observe that our slight modification of the cost of arc $(s_0,t_0)$ in both gadgets $N_{n}^{\vec{a}}$ and $N_{n}^{-\vec{a}}$ does not affect the behavior of the \SSPA{}. 
This is because the cost of any path in $G$ is perturbed by at most $\frac25 \varepsilon$, and hence the shortest path remains the same in every iteration. The only purpose of the modification is tie-breaking.
  
Consider the behavior of the \SSPA{} on the network $G_{\mathrm{ssp}}^{\vec{a}}$ with arc~$e$ removed.
 In each iteration, the shortest $s$-$t$-path goes via one of the two gadgets.
 By Lemma~\ref{lem:ssp_gadget}, each gadget can be in one of $2^n+1$ states and we number these states increasingly from $0$ to $2^n$ by the order of their appearance during the execution of the \SSPA{}.
 The shortest $s$-$t$-path through either gadget in state $j=0,\dots,2^n-1$ has a cost in the range $[j-A,j+A]$, and hence it is cheaper to use a gadget in state $j$ than the other gadget in state $j+1$.
  This means that after every two iterations both gadgets are in the same state.
  
  Now consider the network $G_{\mathrm{ssp}}^{\vec{a}}$ with arc $e$ put back.
  We show that, as before, if the two gadgets are in the same state before iteration $2j$, $j=0,\dots,2^n-1$, then they are again in the same state two iterations later.
  More importantly, arc $e$ is used in iterations $2j$ and $2j+1$ if and only if $\vec{a}_n^{[j]} = 0$.
  This proves the lemma since, by Proposition~\ref{prop:ak}, $\vec{a}_n^{[j]} = 0$ for some $j < 2^n$ if and only if the \noun{Partition} instance $\vec{a}$ has a solution.
  
  To prove our claim, assume that both gadgets are in the same state before iteration $2j$.
  Let $P^+$ be the shortest $s$-$t$-path that does not use any arc of $N_n^{-\vec{a}}$,  $P^-$ be the shortest $s$-$t$-path that does not use any arc of $N_n^{\vec{a}}$, and $P$ be the shortest $s$-$t$-path using arc $e$.
  Note that one of these paths is the overall shortest {$s$-$t$-path}.
  We distinguish two cases, depending on whether the arc $(s_0,t_0)$ currently carries flow~$0$ or~$1$ in both gadgets.
  
  If $(s_0,t_0)$ carries flow~$0$, then $P^+$, $P^-$ use arc $(s_0,t_0)$ in forward direction.
    Therefore, by Lemma~\ref{lem:ssp_gadget}, the cost of~$P^+$ is~$j + \vec{a}_n^{[j]} + \frac15 \varepsilon$, while the cost of~$P^-$ is~$j - \vec{a}_n^{[j]} + \frac15 \varepsilon$.
  On the other hand, path $P$ follows $P^+$ to node $s_0$ of $N_n^{\vec{a}}$, then uses arc $e$, and finally follows $P^-$ to $t$.
  The cost of this path is exactly~$j$.
  If $\vec{a}_n^{[j]} \neq 0$, then one of $P^+$, $P^-$ is cheaper than $P$, and the next two iterations augment flow along paths $P^+$ and $P^-$.
  Otherwise, if $\vec{a}_n^{[j]}=0$, then $P$ is the shortest path, followed in the next iteration by the path from $s$ to node $t_0$ of~$N_n^{-\vec{a}}$ along $P^-$, along arc $e$ in backwards direction to node $s_0$ of $N_n^{\vec{a}}$, and finally to $t$ along $P^+$, for a total cost of $j + \frac25 \varepsilon$. 
  
  If $(s_0,t_0)$ carries flow~$1$, then $P^+$, $P^-$ use arc $(s_0,t_0)$ in backward direction.
  By Lemma~\ref{lem:ssp_gadget}, the cost of $P^+$ is $j + \vec{a}_n^{[j]} - \frac15 \varepsilon$, while the cost of $P^-$ is $j - \vec{a}_n^{[j]} - \frac15 \varepsilon$.
  On the other hand, path $P$ follows $P^+$ to node $s_0$ of $N_n^{\vec{a}}$, then uses arc $e$, and finally follows $P^-$ to $t$.
  The cost of this path is $j - \frac25 \varepsilon$.
  If $\vec{a}_n^{[j]} \neq 0$, then one of $P^+$, $P^-$ is cheaper than $P$, and the next two iterations augment flow along paths $P^+$ and $P^-$.
  Otherwise, if $\vec{a}_n^{[j]}=0$, then $P$ is the shortest path, followed in the next iteration by the path from $s$ to node $t_0$ of $N_n^{-\vec{a}}$ along $P^-$, along arc~$e$ in backwards direction to node $s_0$ of $N_n^{\vec{a}}$, and finally to $t$ along $P^+$, for a total cost of $j$.
\end{proof}

\section{Omitted proofs of Section~\ref{sec:ns}}

\setcounter{lemma}{2}
\begin{lemma}
Let $S_{i}^{\vec{v},r}$, $\vec{v}\in\{\vec{a},-\vec{a}\}$,
be part of a larger network $G$ and assume that before every iteration
of the \NSA{} on $G$ where flow is routed through $S_{i}^{\vec{v},r}$
there is a tree-path from $t_{i}$ to $s_{i}$ in the residual network
of $G$ that has cost smaller than $-2^{i+1}$ and capacity greater than~$1$. Then, there are exactly $2x_{i}=3\cdot2^{i}-2$ iterations in which
one unit of flow is routed from $s_{i}$ to~$t_{i}$ along arcs of $S_{i}^{\vec{v},r}$.
Moreover:
\begin{enumerate}
\item In iteration $j=3k$, $k=0,\dots,2^{i}-1$, arc $(s_{0},t_{0})$ enters
the basis carrying flow $k\,\mathrm{mod}\,2$ and immediately exits the basis again
carrying flow $(k+1)\,\mathrm{mod}\,2$. The cost incurred by arcs of
$S_{i}^{\vec{v},r}$ is $k+\vec{v}_{i}^{[k]}$.
\item In iterations $j=3k+1,3k+2$, $k=0,\dots,2^{i}-2$, for some $0\leq i'\leq i$, the cost
incurred by arcs of $S_{i}^{\vec{v},r}$ is $k+r+\vec{v}_{i',i}^{[k]}$
and $k+(1-r)+\vec{v}_{i',i}^{[k]}$ in order of increasing cost. One
of the arcs $(s_{i'},s_{i'-1}),(s_{i'-1},t_{i'})$ and one of the arcs
$(s_{i'},t_{i'-1}),(t_{i'-1},t_{i'})$ each enter and leave the basis
in these iterations. 
\end{enumerate}
\end{lemma}

\begin{proof}
First observe that throughout the execution of the \NSA{} on~$G$, one unit of flow must always be routed along both of the paths $s_{i},s_{i-1},\dots,s_{0}$ and $t_{0},t_{1},\dots,t_{i}$.
This is because there is an initial flow of one along these paths, all of $s_{0},\dots,s_{n-1}$ have in-degree~$1$, and all of $t_{0},\dots,t_{n-1}$ have out-degree~$1$, which means that the flow cannot be rerouted. 

We prove the lemma by induction on $i>0$, together with the additional
property, that after $2x_{i}$ iterations the arcs in $S_{i-1}^{\vec{v},r}$
carry their initial flow values, while the arcs in $S_{i}^{\vec{v},r}\setminus S_{i-1}^{\vec{v},r}$
all carry $x_{i}$ additional units of flow (which implies that $(s_{i},s_{i-1})$
and $(t_{i-1},t_{i})$ are saturated). Also, the configuration of
the basis is identical to the initial configuration, except that the
membership in the basis of arcs in $S_{i}^{\vec{v},r}\setminus S_{i-1}^{\vec{v},r}$
is inverted. In the following, we assume that $r\in(2A,1/2)$,
the case where $r\in(1/2,1-2A)$ is analogous. In each iteration $j$,
let $P_{j}$ denote the tree-path outside of~$S_{i}^{\vec{v},r}$ from $t_{i}$ to $s_{i}$ of cost
$c_{j}<-2^{i+1}$ and capacity greater than~$1$.

For $i=1$, the \NSA{} performs the following four iterations
involving $S_{1}^{\vec{v},r}$ (cf.~Figure~\ref{fig:ssp_example} for and illustration embedded in $S_{2}^{\vec{v},r}$). In the first iteration, $(s_{0},t_{0})$
enters the basis and one unit of flow is routed along the cycle $s_{1},s_{0},t_{0},t_{1},P_{0}$
of cost $v_{1}+c_{0}=\vec{v}_{1}^{[0]}+c_{0}$. This saturates
arc $(s_{0},t_{0})$ which is the unique arc to become tight (since
$P_{0}$ has capacity greater than~$1$) and thus exits the basis again. In the second
iteration, $(s_{0},t_{1})$ enters the basis and one unit of flow
is routed along the cycle $s_{1},s_{0},t_{1},P_{1}$ of cost $r+c_{1}=r+\vec{v}_{1,1}^{[0]}+c_{1}$,
thus saturating (together with the initial flow of~$1$) arc $(s_{0},s_{1})$
of capacity $x_{1}+1=3$. Since $P_{1}$ has capacity greater than~$1$, this
is the only arc to become tight and it thus exits the basis. In the
third iteration, $(s_1,t_{0})$ enters the basis and one unit of
flow is routed along the cycle $s_{1},t_{0},t_{1},P_{2}$ of cost $(1-r)+c_{2}=(1-r)+\vec{v}_{1,1}^{[0]}+c_{2}$.
Similar to before, $(t_{0},t_{1})$ is the only arc to become tight and thus
exits the basis. In the fourth and final iteration, $(s_{0},t_{0})$ enters
the basis and one unit of flow is routed along the cycle $s_{1},t_{0},s_{0},t_{1},P_{3}$
of cost $1-v_{1}+c_{3}=\vec{v}_{1}^{[1]}+c_{3}$, which causes $(s_{0},t_{0})$
to become empty and leave the basis. Thus, after four iterations, arc $(s_0,t_0)$
in $S_{0}^{\vec{v},r}$ carries its initial flow of value~$0$, while the arcs
in $S_{1}^{\vec{v},r}\setminus S_{0}^{\vec{v},r}$ all carry $2=x_{1}$
additional units of flow. Also, the arcs $(s_{0},t_{1}),(s_{1},t_{0})$
replaced the arcs $(s_{1},s_{0}),(t_{0},t_{1})$ in the basis. 

To see, for $i>0$, that $S_{i}^{\vec{v},r}$ is saturated after $2x_{i}$
units of flow have been routed from $s_{i}$ to $t_{i}$, consider
the directed $s_{i}$-$t_{i}$-cut in $S_{i}^{\vec{v},r}$ induced
by $\{s_{i},t_{i-1}\}$ containing the arcs $(s_{i},s_{i-1})$,
$(t_{i-1},t_{i})$. The capacity of this cut is exactly $2x_{i}+2$
and the initial flow over the cut is~$2$.

Now assume our claim holds for $S_{i-1}^{\vec{v},r}$ and $S_{i-1}^{\vec{v},1-r}$
and consider $S_{i}^{\vec{v},r}$. Consider the first $2x_{i-1}$ iterations $j=0,\dots,2x_{i-1}-1$ and set
$k:=\left\lfloor j/3\right\rfloor <2^{i-1}$. It can be seen inductively that
the shortest path from $t_{i-1}$ to $s_{i-1}$ in the bidirected network associated
with $S_{i-1}^{\vec{v},r}$ has cost at least $-2^{i-1}+1-A>-2^{i-1}+1-r$. 
Hence, every path from $s_{i}$ to~$t_{i}$ using either or both of the arcs $(s_{i},t_{i-1})$ or $(s_{i-1},t_{i})$ has cost greater than $2^{i-1}-(1-r)-A>2^{i-1}-1+A$.
By induction, we can thus infer that none of these arcs
enters the basis in iterations $j<2x_{i-1}$, and instead an arc of
$S_{i-1}^{\vec{v},r}$ enters (and exits) the basis and one unit of
flow gets routed from $s_{i}$ to $t_{i}$ via the arcs $(s_{i},s_{i-1})$,
$(t_{i-1},t_{i})$. We may use induction here since, before iteration $j$, 
the path $t_{i-1},t_{i},P_{j},s_{i},s_{i-1}$
has cost $v_{i}+c_{j}<v_{i}-2^{i+1}<-2^{i}$ and its capacity is greater
than~$1$, since both $(s_{i},s_{i-1})$, $(t_{i-1},t_{i})$
have capacity $x_{i}+1=2x_{i-1}+2$, leaving one unit of spare capacity
even after a flow of $2x_{i-1}$ has been routed along them in addition
to the initial unit of flow. The additional cost contributed by arcs
$(s_{i},s_{i-1}),(t_{i-1},t_{i})$ is $v_{i}$, which is in accordance
with our claim since $\vec{v}_{\ell,i-1}^{[k]}+v_{i}=\vec{v}_{\ell,i}^{[k]}$
for all $\ell\in\{0,\dots,i-1\}$ and $k\in\{0,\dots,2^{i-1}-1\}$.

Because $S_{i-1}^{\vec{v},r}$ is fully saturated after $2x_{i-1}$
iterations, in the next iteration $j=2x_{i-1}=3\cdot2^{i-1}-2$, $k:=\left\lfloor j/3\right\rfloor =2^{i-1}-1$,
arc $(s_{i-1},t_{i})$ is added to the basis and one unit of flow is sent
along the path $s_{i},s_{i-1},t_{i}$, thus saturating the capacity
$x_{i}+1=2x_{i-1}+2$ of arc $(s_{i},s_{i-1})$ and incurring a
cost of $2^{i-1}-(1-r)=k+r+\vec{v}_{i,i}^{[k]}$. Note that this cost
is higher than the cost of each of the previous iterations. The saturated
arc has to exit the basis since, by assumption, $P_{j}$ has capacity
greater than~$1$. Similarly, in the following iteration $j=2x_{i-1}+1=3\cdot2^{i-1}-1$,
$k:=\left\lfloor j/3\right\rfloor =2^{i-1}-1$, the cost is $2^{i-1}-r=k+(1-r)+\vec{v}_{i,i}^{[k]}$
and arc $(t_{i-1},t_{i})$ is replaced by $(s_{i},t_{i-1})$ in the basis.

By induction, at this point $(s_{i-2},t_{i-1})$ and $(s_{i-1},t_{i-2})$
are in the basis, the arcs of $S_{i-1}^{\vec{v},r}\setminus S_{i-2}^{\vec{v},r}$
carry a flow of $x_{i-1}$ in addition to their initial flow, and
$S_{i-2}^{\vec{v},r}$ is back to its initial configuration. To be able to apply induction on the residual network of $S_{i-1}^{\vec{v},r}$, we shift
the costs of the arcs at $s_{i-1}$ by $-(2^{i-2}-r)$ and the costs
of the arcs at $t_{i-1}$ by $-(2^{i-2}-(1-r))$ in the residual network
of $S_{i-1}^{\vec{v},r}$. Since we shift costs uniformly across cuts,
this only affects the costs of paths but not the structural behavior
of the gadget. Specifically, the costs of all paths from $t_{i-1}$ to
$s_{i-1}$ in the residual network are increased by exactly $2^{i-1}-1$.
If we switch roles of $s_{i-1}$ and $t_{i-1}$, say $\tilde{s}_{i-1}:=t_{i-1}$
and $\tilde{t}_{i-1}:=s_{i-1}$, we obtain the residual network
of $S_{i-1}^{\vec{v},1-r}$ with its initial flow. This allows us to
use induction again for the next $2x_{i-1}$ iterations.

To apply the induction hypothesis, we need the tree-path from
$\tilde{t}_{i-1}=s_{i-1}$ to $\tilde{s}_{i-1}=t_{i-1}$ to maintain
cost smaller than $-2^{i}$ and capacity greater than~$1$. This is
fulfilled since $P_{j}$ has cost smaller than $-2^{i+1}$, which is
sufficient even with the additional cost of $2^{i}-1-v_{i}$ incurred
by arcs $(s_{i},\tilde{s}_{i-1})$, $(\tilde{t}_{i-1},t_{i})$. The
residual capacity of $(t_{i},\tilde{t}_{i-1})$ and $(\tilde{s}_{i-1},s_{i})$
is $x_{i}>2x_{i-1}$ and thus sufficient as well. By induction for
$S_{i-1}^{\vec{v},1-r}$, we may thus conclude that in iterations $j=2x_{i-1}+2,\dots,2x_{i}-1$, $k:=\left\lfloor j/3\right\rfloor \geq2^{i-1}$,
one unit of flow is routed via $(s_{i},t_{i-1}),S_{i-1}^{\vec{v},r},(s_{i-1},t_{i})$.
The cost of $(s_{i},\tilde{s}_{i-1})$ and $(\tilde{t}_{i-1},t_{i})$
together is $2^{i}-1-v_{i}$. The cost of iteration $j'=j-2x_{i-1}-2$,
$k':=\left\lfloor j'/3\right\rfloor =k-2^{i-1}$, in $S_{i-1}^{\vec{v},1-r}$
is $k'+y+\vec{v}_{\ell,i-1}^{[k']}$, for $y\in\{0,r,(1-r)\}$ and $\ell\in\{0,\dots,i-1\}$
chosen according to the different cases of the lemma. Accounting for
the shift by $2^{i-1}-1$ of the cost compared with the residual network
of $S_{i-1}^{\vec{v},r}$, the incurred total cost in $S_{i-1}^{\vec{v},r}$
is 
\begin{align*}
 (2^{i}-1-v_{i})&+(k'+y+\vec{v}_{\ell,i-1}^{[k']})-(2^{i-1}-1)\\
 & =2^{i-1}+k'+y-v_{i}+\vec{v}_{\ell,i-1}^{[k']} 
 =k+y+\vec{v}_{\ell,i}^{[k]},
\end{align*}
where we used $-v_{i}+\vec{v}_{\ell,i-1}^{[k']}=\vec{v}_{\ell,i}^{[k'+2^{i-1}]}$
since $k'<2^{i-1}$. This concludes the proof.
\end{proof}

\setcounter{lemma}{3}

\begin{lemma}
Arc $e$ enters the basis in some iteration of the \NSA{} on network $G_{\mathrm{ns}}^{\vec{a}}$ if and only if the \noun{Partition} instance $\vec{a}$ has a solution.
\end{lemma}

\begin{proof}
First observe that $\vec{a}_n^{[2k]}=\vec{a}_n^{[2k+1]}$ for $k\in {0,\dots,2^{n-1}}$ since, by assumption, $a_1 = 0$.
  
Similar to the proof of Lemma~\ref{lem:ssp_construction}, in isolation each of the two gadgets can be in one of $2x_n$ states (Lemma~\ref{lem:ns_gadget}), which we label by the number of iterations needed to reach each state.
  Assuming that both gadgets are in state $12k$ after some number of iterations, we show that both gadgets will reach state $12k+12$ together as well.
  In addition, we show that, in the iterations in-between, arc $e$ enters the basis if and only if $\vec{a}_n^{[4k]} = 0$ and thus $\vec{a}_n^{[4k+1]} = 0$, or $\vec{a}_n^{[4k+2]} = 0$ and thus $\vec{a}_n^{[4k+3]} = 0$.
  Consider the situation where both gadgets are in state $12k$.
  Note that in this state the arcs in $S_1^{\vec{v},1/3}$ and $S_1^{-\vec{v},1/3}$ are back in their original configuration.
  
  Let $P^{\pm}$ denote the tree-path from $t_1^{\pm}$ to $s_1^{\pm}$, and let $P^{\pm\mp}$ denote the tree-path from $t_1^{\mp}$ to $s_1^{\pm}$.
  We refer to these paths as the \emph{outer} paths.
  Observe that, since the gadgets are in the same state, the costs of the outer paths differ by at most $A<1/4$.
  In the next iterations, flow is sent along a cycle containing one of the outer paths, and we analyze only the part of each cycle without the outer path.
  Let $P_0^{\pm},P_1^{\pm},P_2^{\pm},P_3^{\pm}$ be the four successive shortest paths within the gadget $S_1^{\pm\vec{a},1/3}$.
  The costs of these paths are $\frac15\eps$, $1/3$, $2/3$, $1-\frac15\eps$, respectively. 
  Note that, since $A < 1/6$, the costs of the paths stay in the same relative order within each gadget throughout the algorithm. 
  
  If $\vec{a}_n^{[4k]} < 0$, then $P^+$ is the cheapest of the outer paths by a margin of more than $\varepsilon / 2$.
  Thus, in the first iteration, $(c^+,t_0^+)$ replaces $(s_0^+,c^+)$ in the basis closing the path $P_0^+$.
  In the next five iterations, the paths $P_0^-$, $P_1^+$, $P_1^-$, $P_2^+$, $P_2^-$ are closed in this order.
  The final two iterations are $P_3^+$, $P_3^-$, similar to the first two iterations, as $\vec{a}_n^{[4k + 1]} = \vec{a}_n^{[4k]} < 0$.
  At this point, $8$ iterations have passed and both gadgets are in state $12k+6$.
  
  If $\vec{a}_n^{[4k]} > 0$, then $P^-$ is the cheapest of the outer paths by a margin of more than $\varepsilon / 2$.
  Thus, the first iteration closes the path $P_0^-$.
  The next five iterations are via $P_0^+$, $P_1^{-}$, $P_1^{+}$, $P_2^-$, $P_2^+$, in this order.
  The final two iterations are $P_3^-$, $P_3^+$, similar to the first two iterations, as $\vec{a}_n^{[4k + 1]} = \vec{a}_n^{[4k]} > 0$.
  At this point, $8$ iterations have passed and both gadgets are in state $12k+6$.
  
  If $\vec{a}_n^{[4k]} = 0$, then all four outer paths have the same cost.
  The first iteration is via the path $s_1^+,s_0^+,c^+,c^-,t_0^-,t_1^-$, \ie, arc $e$ enters and leaves the basis, for a cost of $0$ and an additional flow of $1/2$.
  The next two iterations are via $P_1^{\pm}$, each for a cost of $\frac15\eps$ and an additional flow of $1/2$. 
  The fourth iteration is via the path $s_1^-,s_0^-,c^-,c^+,t_0^+,t_1^+$, \ie, arc $e$ enters and leaves the basis again, for a cost of $\frac25\eps$ and an additional flow of $1/2$.
  The next iterations are as before: via $P_1^+$, $P_1^-$, $P_2^-$, $P_2^+$, in this order.
  The final four iterations are similar to the first four iterations, again twice using $e$, as $\vec{a}_n^{[4k + 1]} = \vec{a}_n^{[4k]} = 0$.
  At this point, $12$ iterations have passed and both gadgets are in state $12k+6$.
  
  The next four iterations (two for each gadget) do not involve the subnetworks $S_1^{\vec{a},1/3}$ and $S_1^{-\vec{a},1/3}$, and do thus not use $e$.
  The iterations going from state $12k+6$ to state $12k+12$ are analogous to the above if we exchange the roles of $s_1^{\pm}$ and~$t_1^{\pm}$. This concludes the proof.
\end{proof}

\section{Omitted proofs of Corollaries}\label{sec:corollaries}

\setcounter{corollary}{0}

\begin{corollary}
Determining the number of iterations needed by the \SA{}, the \NSA{}, and the \SSPA{} for a given input is \noun{NP}-hard.
\end{corollary}

\begin{proof}
We first show that determining the number of iterations needed by the \SSPA{} for a given minimum-cost flow instance is \noun{NP}-hard. We replace the arc $e$ in $G_{\mathrm{ssp}}^{\vec{a}}$ of Section~\ref{sec:ssp} by two parallel arcs, each with a capacity of $1/2$ and slightly perturbed costs. This way, every execution of the \SSPA{} that previously did not use arc $e$ is unaffected, while executions using $e$ require additional iterations. Thus, by Lemma~\ref{lem:ssp_construction}, the \SSPA{} on network $G_{\mathrm{ssp}}^{\vec{a}}$ takes more than $2^{n+1}$ iterations if and only if the \noun{Partition} instance $\vec{a}$ has a solution.

The proof for the \NSA{} (and thus the \SA{}) follows from the proof of Lemma~\ref{lem:ns_construction}, observing that the \NSA{} takes more than $4x_n$ iterations for network $G_{\mathrm{ns}}^{\vec{a}}$ if and only if the \noun{Partition} instance $\vec{a}$ has a solution.
\end{proof}

\begin{corollary}
Deciding for a given linear program whether a given variable ever enters the basis during the execution of the \SA{} is \noun{NP}-hard.
\end{corollary}

\begin{proof}
The proof is immediate via Lemma~\ref{lem:ns_construction} and the fact that \noun{Partition} is \noun{NP}-hard.
\end{proof}

\begin{corollary}
Determining whether a parametric minimum-cost flow uses a given arc (\ie, assigns positive flow value for any parameter value) is \noun{NP}-hard. In particular, determining whether the solution to a parametric linear program uses a given variable is \noun{NP}-hard. Also, determining the number of different basic solutions over all parameter values is \noun{NP}-hard.
\end{corollary}

\begin{proof}
This follows from the fact that the \SSPA{} solves a parametric minimum-cost flow problem, together with Lemma~\ref{lem:ssp_construction} and Corollary~\ref{cor:ssp_iterations}.
\end{proof}

\begin{corollary}
Given a $d$-dimensional polytope~$P$ by a system of linear inequalities, determining the number of vertices of $P$'s projection onto a given $2$-dimensional subspace is \noun{NP}-hard.
\end{corollary}

\begin{proof}
Let $P$ be the polytope of all feasible $s$-$t$-flows in network $G_{\mathrm{ssp}}^{\vec{a}}$ of Section~\ref{sub:ssp_construction}.
Consider the $2$-dimensional subspace $S$ defined by flow value and cost of a flow.
Let $P'$ be the projection of $P$ onto $S$.
The lower envelope of $P'$ is the parametric minimum-cost flow curve for $G_{\mathrm{ssp}}^{\vec{a}}$, while the upper envelope is the parametric maximum-cost flow curve for $G_{\mathrm{ssp}}^{\vec{a}}$.

The $s$-$t$-paths of maximum cost in $G_{\mathrm{ssp}}^{\vec{a}}$ are the four paths via $s_n,s_{n-1},t_n$ or via $s_n,t_{n-1},t_n$ in both of the gadgets.
Each of these paths has cost $2^{n-1} - \frac12$ and the total capacity of all paths together is $2^{n+1}$ which is equal to the maximum flow value from $s$ to $t$.
Therefore, the upper envelope of $P'$ consists of a single edge.

The number of edges on the lower envelope of $P'$ is equal to the number of different costs among all successive shortest paths in $G_{\mathrm{ssp}}^{\vec{a}}$.
If we slightly perturb the costs of the two arcs in $G_{\mathrm{ssp}}^{\vec{a}}$ with cost $\frac15\eps$, we can ensure that each successive shortest path has a unique cost.
The claim then follows by Corollary~\ref{cor:ssp_iterations}.
\end{proof}

\begin{corollary}
Determining the average arrival time of flow in an earliest arrival flow is \noun{NP}-hard. 
\end{corollary}

\begin{proof}[Sketch]
The average arrival time can be obtained from the parametric minimum-cost flow curve considered in the proof of Corollary~\ref{cor:projection}. By slightly perturbing the cost of arc~$e$ in network $G_{\mathrm{ssp}}^{\vec{a}}$, the value of the average arrival time discloses whether~$e$ is used by the \SSPA{}. The result thus follows from Lemma~\ref{lem:ssp_construction}.
\end{proof}

\end{document}